\documentclass[11pt]{article}
\usepackage[a4paper,margin=1in]{geometry}
\usepackage{amsmath,amssymb,amsthm}
\usepackage{mathtools}
\usepackage[hidelinks]{hyperref}
\usepackage{enumitem}
\usepackage{booktabs}
\usepackage{orcidlink}
\theoremstyle{plain}
\newtheorem{theorem}{Theorem}
\newtheorem{proposition}{Proposition}
\newtheorem{corollary}{Corollary}

\theoremstyle{definition}
\newtheorem{definition}{Definition}
\theoremstyle{remark}
\newtheorem{remark}{Remark}
\renewcommand{\Pr}{\mathbb{P}}
\newcommand{\Hent}{H}
\newcommand{\Zmod}[1]{\mathbb{Z}_{#1}}
\newcommand{\keyspace}{\mathcal{K}}
\title{New Ideas on a New Old Type of Cipher:\\
The Mixed-Radix One-Time Pad}
\author{Fabio F.G. Buono~\orcidlink{0009-0004-9199-2793} \\ Independent Researcher}
\date{\today}
\begin{document}
\maketitle
\begin{abstract}
In a short 2012 preprint, an unconventional cipher was introduced, now in this note we take
that representational core, formalize it, and use it as the basis for a
clean generalization of the one-time pad to non-uniform bases, which we call the
\emph{Mixed-Radix One-Time Pad} (MR-OTP). And we prove that the MR-OTP achieves Shannon
perfect secrecy, show that the classical binary OTP is exactly the all-bases-equal-2
special case, and that fixed-base variants recover OTPs over arbitrary alphabets. We
then examine whether secret bases can lower the key entropy required for perfect secrecy 
(they cannot). We close with a usable session protocol based on key rolling that preserves
perfect secrecy, and with an honest account of the open problems.
\end{abstract}
\begin{center}
  \small Preprint
\end{center}
\section{Introduction}
Everything starts from the preprint \emph{A New Type of Cipher}~\cite{buono2012}, its underlying
idea is that an integer can be decomposed, by repeated Euclidean-style division against a 
chosen sequence of decreasing divisors, into a sequence of quotients. This is, in essence, 
a representation of the integer in a system of \emph{mixed bases} (a mixed-radix numeral system). 
The preprint remained an undeveloped sketch and, to our knowledge, has not been cited.

In this note we extract it, formalize it cleanly, and build on it. Our contributions are:

\begin{enumerate}[label=(\roman*)]
  \item We formalize the mixed-radix representation underlying~\cite{buono2012} and
        define the \emph{Mixed-Radix One-Time Pad} (MR-OTP), the natural one-time pad
        over a mixed-radix digit system (Section~\ref{sec:mrotp}).
  \item We prove that the MR-OTP satisfies Shannon perfect secrecy, and that it is
        correct (i.e., decryption is the exact inverse of encryption). We show the
        classical binary OTP is the special
        case where every base equals~2, and that any constant base $b$ recovers the OTP
        over an alphabet of size~$b$ (Section~\ref{sec:secrecy}).
  \item We state and resolve a tempting conjecture: that \emph{secret} bases might let
        the message-covering key be shorter than the message while retaining perfect
        secrecy. We show this is false in the strong (information-theoretic) sense, and
        explain precisely why, it would break the uniformity hypothesis that the
        secrecy proof requires (Section~\ref{sec:entropy}).
  \item We describe a usable session protocol based on key rolling that preserves
        perfect secrecy, in which the mixed bases provide encoding efficiency and
        adaptivity rather than a reduction in key length (Section~\ref{sec:protocol}).
\end{enumerate}

\section{The Mixed-Radix Representation and the MR-OTP}
\label{sec:mrotp}

\subsection{Mixed-radix representation}
Fix a length $L$ and a sequence of bases $B=(b_1,\dots,b_L)$ with each $b_i\ge 2$. The
\emph{digit space} is
\[
  \mathcal{D}_B \;=\; \prod_{i=1}^{L}\{0,1,\dots,b_i-1\}.
\]
A message $M$ is identified with its digit tuple $(m_1,\dots,m_L)\in\mathcal{D}_B$,
where $0\le m_i<b_i$. With positional weights
\[
  W_i \;=\; \prod_{j>i} b_j, \qquad W_L=1,
\]
the tuple corresponds to the integer $\sum_{i=1}^L m_i W_i$, and conversely every
integer in $\{0,\dots,\prod_i b_i-1\}$ has a unique such representation. This is the
mixed-radix numeral system, and it is the ``cleaned-up'' form of the decreasing-divisor
decomposition of~\cite{buono2012}, with the bases playing the role of the divisors but
without the side constraints (decreasing order, remainder handling) that made the
original sketch awkward to analyze.

\subsection{The cipher}

\begin{definition}[MR-OTP]
Let $B=(b_1,\dots,b_L)$ be a fixed sequence of bases with $b_i\ge2$. The key space is
\[
  \keyspace \;=\; \prod_{i=1}^{L}\{0,1,\dots,b_i-1\},
\]
and a key $K=(k_1,\dots,k_L)$ is drawn uniformly at random from $\keyspace$.
Encryption of a message $M=(m_1,\dots,m_L)$ is performed digitwise:
\[
  c_i \;=\; (m_i + k_i)\bmod b_i, \qquad C=(c_1,\dots,c_L).
\]
Decryption, given $K$, is
\[
  m_i \;=\; (c_i - k_i)\bmod b_i.
\]
\end{definition}

So the construction is the obvious analogue of the binary OTP, with bitwise XOR replaced
by digitwise addition modulo $b_i$ in each (possibly distinct) base.

\section{Perfect Secrecy and Correctness}
\label{sec:secrecy}

\subsection{Correctness}

\begin{proposition}[Correctness]
For every key $K\in\keyspace$ and every ciphertext $C$, decryption returns the unique
message $M$ that is encrypted.
\end{proposition}

\begin{proof}
For each position $i$, the map $x\mapsto (x+k_i)\bmod b_i$ is a translation on the
finite cyclic group $\Zmod{b_i}$, hence a bijection, with inverse
$x\mapsto (x-k_i)\bmod b_i$. Therefore $c_i$ determines $m_i$ uniquely given $k_i$, and
the tuple $C$ determines $M$ uniquely given $K$.
\end{proof}

\subsection{Perfect secrecy}

We use Shannon's definition~\cite{shannon1949}.

\begin{definition}[Perfect secrecy]
A cipher has \emph{perfect secrecy} if for every message distribution $\Pr[M]$ and
every ciphertext $c$ with $\Pr[C=c]>0$,
\[
  \Pr[M=m\mid C=c] \;=\; \Pr[M=m] \qquad \text{for all } m.
\]
\end{definition}

\begin{theorem}[Perfect secrecy of the MR-OTP]
\label{thm:secrecy}
The MR-OTP with a uniformly random key over $\keyspace$ has perfect secrecy.
\end{theorem}

\begin{proof}
We compute the three quantities of Bayes' rule.

\emph{Step 1: the conditional $\Pr[C=c\mid M=m]$.}
Fix a position $i$ and fix $m_i,c_i$. The equation $c_i=(m_i+k_i)\bmod b_i$ has exactly
one solution $k_i\equiv (c_i-m_i)\pmod{b_i}$ in $\{0,\dots,b_i-1\}$. Since $k_i$ is
uniform on $\{0,\dots,b_i-1\}$,
\[
  \Pr[c_i\mid m_i] \;=\; \frac{1}{b_i}.
\]
The key digits $k_1,\dots,k_L$ are chosen \emph{independently}, so the events across
positions are independent given $M$, and
\[
  \Pr[C=c\mid M=m] \;=\; \prod_{i=1}^{L}\Pr[c_i\mid m_i] \;=\; \prod_{i=1}^{L}\frac{1}{b_i},
\]
which does not depend on $m$.

\emph{Step 2: the marginal $\Pr[C=c]$.}
By the law of total probability and Step~1,
\[
  \Pr[C=c] \;=\; \sum_{m'} \Pr[C=c\mid M=m']\,\Pr[M=m']
           \;=\; \Big(\prod_{i=1}^{L}\tfrac{1}{b_i}\Big)\sum_{m'}\Pr[M=m']
           \;=\; \prod_{i=1}^{L}\frac{1}{b_i},
\]
since $\sum_{m'}\Pr[M=m']=1$.

\emph{Step 3: the posterior.} By Bayes' rule,
\[
  \Pr[M=m\mid C=c]
  \;=\; \frac{\Pr[C=c\mid M=m]\,\Pr[M=m]}{\Pr[C=c]}
  \;=\; \frac{\big(\prod_i \tfrac{1}{b_i}\big)\Pr[M=m]}{\prod_i \tfrac{1}{b_i}}
  \;=\; \Pr[M=m],
\]
for all $m,c$. This is perfect secrecy.
\end{proof}

\begin{remark}
About the independence of the key digits, invoked in Step~1, it is what makes
the per-position factor $1/b_i$ multiply into a message-independent product. The
uniformity of each $k_i$ over the \emph{entire} set $\{0,\dots,b_i-1\}$ is equally
essential; we return to this point in Section~\ref{sec:entropy}.
\end{remark}

\subsection{The OTP as a special case}

\begin{corollary}[Binary OTP]
\label{cor:binary}
With $b_i=2$ for all $i$, the MR-OTP coincides exactly with the classical binary
one-time pad.
\end{corollary}

\begin{proof}
With $b_i=2$ we have $m_i,k_i\in\{0,1\}$ and
$c_i=(m_i+k_i)\bmod 2 = m_i \oplus k_i$, addition modulo~2 being exactly XOR. Hence
$C=M\oplus K$ bitwise, the definition of the binary OTP.
\end{proof}

\begin{corollary}[Single-alphabet OTP]
\label{cor:alphabet}
With $b_i=b$ constant for all $i$, the MR-OTP is the one-time pad over an alphabet of
size $b$ (e.g.\ $b=26$ for the Latin alphabet, $b=4$ for a four-symbol alphabet such as
nucleotide data).
\end{corollary}

\begin{proof}
With $b_i=b$ for all $i$, the digit space becomes $\{0,\dots,b-1\}^L$, encryption
is $c_i=(m_i+k_i)\bmod b$, and the cipher is exactly the one-time pad over an
alphabet of size~$b$.
\end{proof}

These corollaries position the MR-OTP as a generalization rather than a rival
of the one-time pad and the classical and single-alphabet pads are recovered as points in
a family, and the truly new territory is that of \emph{mixed} (non-constant) bases,
where the digit system can be matched to the structure of the data being encrypted.

\section{Can Secret Bases Reduce the Required Key Entropy?}
\label{sec:entropy}

A natural hope is that, by treating the bases $B$ as part of the secret key, one might
get away with a message-covering key shorter than the message while keeping perfect
secrecy, the long-standing dream of a ``perfect cipher with short keys.'' We make the
question precise and show the answer is negative.

\subsection{Setup}

Now suppose the key has two parts, $K=(B,\mathbf{k})$, where $B=(b_1,\dots,b_L)$ is now
secret and $\mathbf{k}=(k_1,\dots,k_L)$ is the additive pad with $k_i\in\{0,\dots,b_i-1\}$.
Write $\Hent$ for Shannon entropy. The entropy of the key decomposes as
\[
  \Hent(K) \;=\; \Hent(B) + \Hent(\mathbf{k}\mid B).
\]
Given $B$, and we assuming, as required for perfect secrecy, that $\mathbf{k}$ is uniform
and independent given $B$, a uniform pad contributes
\[
  \Hent(\mathbf{k}\mid B) \;=\; \sum_{i=1}^{L}\log_2 b_i,
\]
while the message, living in $\mathcal{D}_B$, satisfies
\[
  \Hent(M) \;\le\; \sum_{i=1}^{L}\log_2 b_i,
\]
with equality when $M$ is uniform.

\subsection{The accounting identity and what it does \emph{not} buy}

Combining the above,
\[
  \Hent(K) \;=\; \underbrace{\Hent(B)}_{\ge 0} \;+\; \underbrace{\Hent(\mathbf{k}\mid B)}_{\ge \Hent(M)}
           \;\ge\; \Hent(M).
\]
This inequality is consistent with the Shannon bound $\Hent(K)\ge \Hent(M)$ for the perfect
secrecy~\cite{shannon1949}, but it is merely an accounting identity: the pad
$\mathbf{k}$ \emph{already} carries entropy $\sum_i \log_2 b_i \ge \Hent(M)$ on its own,
and the base entropy $\Hent(B)$ sits \emph{on top of} the requirement.

\begin{proposition}[Secret bases cannot substitute for pad entropy]
\label{prop:noshortcut}
Suppose the MR-OTP retains perfect secrecy. Then the pad $\mathbf{k}$ must, on each
position $i$, be uniform over the full set $\{0,\dots,b_i-1\}$, and hence
$\Hent(\mathbf{k}\mid B)\ge \Hent(M)$. In particular, base entropy $\Hent(B)$ cannot be
used to make $\mathbf{k}$ shorter than what is required to cover $\Hent(M)$.
\end{proposition}

\begin{proof}
The secrecy proof of Theorem~\ref{thm:secrecy} requires, at Step~1, that for each $i$
the digit $k_i$ be uniform over the whole of $\{0,\dots,b_i-1\}$. This is what
forces $\Pr[c_i\mid m_i]=1/b_i$ independently of $m_i$. If on some position the pad
fails to range uniformly over $\{0,\dots,b_i-1\}$, which is what ``shortening
$\mathbf{k}$ below the message length'' necessarily entails for at least one
position, then the conditional $\Pr[c_i\mid m_i]$ depends on the $m_i$, the product no
longer factors into a message-independent constant, and the posterior need no longer
equal the prior, thus the perfect secrecy fails. Knowing or hiding $B$ is irrelevant to
this argument, since the failure is on the pad coordinate.
\end{proof} 

The upshot is positive in its own way, as secret bases do not lower the information-theoretic
key requirement, and whatever additional protection they might provide is computational in
nature and outside the scope of this note. So an adversary who does not know $B$ faces a
search problem over the base space before any decryption attempt, and whether this
constitutes meaningful computational security remains an open question.

\section{Perfect Secrecy with Key Rolling as an Usable Protocol}
\label{sec:protocol}

The practical objection to any one-time pad is key distribution. The MR-OTP inherits
this: by Proposition~\ref{prop:noshortcut}, perfect secrecy still demands pad material
at least as long as the message. What the mixed-radix structure adds is not shorter
keys but \emph{encoding efficiency and adaptivity}. We make this precise via a key-rolling
session protocol.

\subsection{Construction}

Alice and Bob share, once and offline, an initial key $K_0$. We \emph{partition}
(never reuse) $K_0$ and its successors. At session $t$, the active key segment is split
as
\[
  K_t \;=\; K_t^{\mathrm{msg}} \,\Vert\, B_{t+1} \,\Vert\, K_{t+1}^{\mathrm{msg}} \,\Vert\, \dots,
\]
where $\Vert$ denotes concatenation of disjoint key material. Session $t$ proceeds as:
\begin{enumerate}[label=\arabic*.]
  \item Encrypt $M_t$ as an MR-OTP under bases $B_t$ and pad $K_t^{\mathrm{msg}}$,
        where $|K_t^{\mathrm{msg}}|$ covers $\Hent(M_t)$ in full.
  \item Read off the next bases $B_{t+1}$ and the next pad segment
        $K_{t+1}^{\mathrm{msg}}$ from disjoint, never-before-used portions of the shared
        key material.
  \item Session $t+1$ uses bases $B_{t+1}$, chosen to match the expected structure of
        $M_{t+1}$.
\end{enumerate}

\subsection{Why this preserves perfect secrecy}

Because every pad segment and the every base segment is taken from disjoint, previously
unused key material, no key bit ever encrypts more than its own entropy's worth of
content. The protocol is thus a partitioned one-time pad, where each session is an
independent MR-OTP to which Theorem~\ref{thm:secrecy} applies directly. The crucial
discipline is that the material used to convey $B_{t+1}$ and $K_{t+1}^{\mathrm{msg}}$
is \emph{never} the same material used to encrypt $M_t$, as reusing it would make some
key segment cover more than its entropy and would break secrecy as in
Proposition~\ref{prop:noshortcut}

\subsection{Role and limits of adaptive bases}

There are two distinct ways adaptive bases might be understood, and conflating them is
the classic error.

\begin{description}[leftmargin=1.5em]
  \item[Encoding efficiency (always safe).] Choosing $B_{t+1}$ to match the alphabet of
        $M_{t+1}$ (e.g.\ base 26 for text, base 4 for nucleotide data) lets the message
        be represented natively, without padding or artificial binary conversion. The
        pad still covers $\Hent(M_{t+1})$ in full; secrecy is unaffected. The gain is
        representational convenience.
  \item[Entropy reduction (conditional, and dangerous).] Choosing $B_{t+1}$ because one
        \emph{assumes} $M_{t+1}$ is non-uniform---and then shortening the pad to the
        assumed $\Hent(M_{t+1})<\sum_i\log_2 b_i$---can be consistent with Shannon
        \emph{only if the message model is exact}. If the real message deviates from the
        assumed model, perfect secrecy is lost. This is precisely the trap that has sunk
        many proposed ``usable perfect'' ciphers.
\end{description}

We therefore recommend, and analyze, adaptive bases strictly in the first sense:
efficiency and flexibility, not a covert reduction of key length.

\subsection{Key consumption and overhead of base transmission}

Because the bases $B_{t+1}$ are themselves transmitted using fresh key material, their
encoding has a cost that must be accounted for. Suppose each base $b_i$ requires at
most $\beta$ bits to represent in the shared key (where $\beta$ depends on the
agreed-upon base alphabet, but is left as a parameter here to keep the result general).
Then the total key material consumed per session $t$ is
\[
  \underbrace{\sum_{i=1}^{L}\log_2 b_i}_{\text{pad}} \;+\; \underbrace{L\beta}_{\text{bases}},
\]
and the overhead fraction relative to the pad is $L\beta / \sum_i \log_2 b_i$. This
overhead is negligible when
\[
  \frac{1}{L}\sum_{i=1}^{L}\log_2 b_i \;\gg\; \beta,
\]
that is, when the average base is large relative to its own representation cost. In
typical cases of interest (e.g.\ bases matched to natural-language or biological
alphabets, where $b_i \ge 4$) this condition holds comfortably. In the degenerate case
$b_i = 2$ for all $i$, the overhead equals the pad length and the mixed-radix structure
provides no representational gain over the binary OTP, as expected.

\section{Connections and Context}
\label{sec:context}

\paragraph{Relation to the OTP over finite abelian groups.}
The key space $\keyspace = \prod_i \Zmod{b_i}$ is a finite abelian group under
componentwise addition, and the MR-OTP is the one-time pad on this group. The
perfect secrecy of OTPs over finite abelian groups is a classical fact, so the
MR-OTP is, algebraically, a special case of a well-known construction. What the
mixed-radix perspective adds is the explicit identification of the group with the
integer representation of a message: the digit structure connects the algebraic
object to the numeral system, and it is this connection that makes the encoding
efficiency of Section~\ref{sec:protocol} and the arithmetic coding analogy below
natural rather than coincidental.

\paragraph{Mixed-radix numeral systems and Diophantine approximation.}
Non-uniform positional systems are classical. The factorial number system ($b_i=i+1$)
is standard~\cite{alloucheshallit2003}, and Ostrowski's numeration~\cite{ostrowski1922}
associates to the continued-fraction expansion of an irrational $\alpha$ a mixed-radix
system whose digit bounds are the partial quotients; for $\alpha$ the golden ratio this
specializes to the Zeckendorf (Fibonacci) representation. Our digit system is a
finite-length, fixed-base instance of this much older and richer theory. The MR-OTP
adds a cryptographic reading of these systems, not new number theory.

\paragraph{Arithmetic coding and entropy.}
The ``adaptive bases'' of Section~\ref{sec:protocol}, in their safe (encoding) sense,
are in the spirit of arithmetic coding, where a message is represented in a numeral
system matched to the source so that the representation length approaches the source
entropy. Rissanen's early work~\cite{rissanen1976} on generalized Kraft inequalities
is a precursor to this line of ideas. Choosing $b_i$ to fit an alphabet is a special,
coarse case of this general technique. The link is real and useful, but it places the
construction inside a mature area of information theory rather than beside it.

\paragraph{What is \emph{not} hidden here.}

Before closing this section it is worth noting that the coincidence of the digit-space 
cardinality $\prod_i b_i$ and the state count of $L$ independent finite systems 
(with $\log\prod_i b_i=\sum_i\log b_i$) is simply the additivity of the logarithm, 
and that perfect secrecy, being information-theoretic, attaches to no complexity-class 
statement.

\section{Conclusion and Open Problems}

We formalized the representational idea of~\cite{buono2012} as the Mixed-Radix One-Time
Pad, proved perfect secrecy and correctness, and showed that secret bases cannot reduce
the information-theoretic key requirement. We gave a key-rolling protocol that keeps
perfect secrecy while using mixed bases for encoding efficiency, and characterized the
key consumption overhead of base transmission. The main open questions are the
average-case hardness of base recovery, optimal base selection for structured sources,
and composability of the rolling protocol across sessions and parties.

\section{Final Remarks}

\begin{remark}

By Shannon's theorem, any perfectly secret cipher must satisfy $|\keyspace| \ge |\mathcal{M}|$.
The MR-OTP meets this condition by construction, since $\keyspace = \mathcal{D}_B = \mathcal{M}$,
so the key space and the message space coincide.
\end{remark}

\begin{remark}
The MR-OTP achieves the tightest possible case of the Shannon's bound, as for any pair
$(M, C)$ there is just one key $K \in \keyspace$, namely $k_i = (c_i - m_i) \bmod b_i$
for each $i$, so no key material is wasted.
\end{remark}

\bibliographystyle{plainurl}
\bibliography{references}

\end{document}